\documentclass{article}
\usepackage[margin=1in]{geometry}
\usepackage[utf8]{inputenc}
\usepackage[T1]{fontenc}

\usepackage{amsmath,amsfonts,amsthm} 
\usepackage{hyperref}
\usepackage{algorithm, algorithmic}
\usepackage{tikz}
\usetikzlibrary{calc}

\newtheorem{theorem}{Theorem}[section]
\newtheorem{definition}[theorem]{Definition}
\newtheorem{lemma}[theorem]{Lemma}

\newcommand{\eps}{\varepsilon}

\newcommand{\ud}{\,\mathrm{d}}

\newcommand\tp{\textsc{TilePacking}}
\newcommand\greedy{\textsc{GreedyPacking}}

\DeclareMathOperator*{\area}{area}
\AtBeginDocument{%
	\DeclareFontShape{T1}{cmr}{m}{scit}{<->ssub*cmr/m/sc}{}%
}

\definecolor{utforest}{RGB}{0,103,90}
\definecolor{utgreen}{RGB}{52,178,53}
\definecolor{utorange}{RGB}{236,122,8}
\definecolor{utyellow}{RGB}{254,209,0}
\definecolor{utnavy}{RGB}{0,44,95}
\definecolor{utblue}{RGB}{99,177,229}
\definecolor{utpurple}{RGB}{79,45,127}
\definecolor{utmagenta}{RGB}{207,0,114}
\definecolor{utgold}{RGB}{136,123,27}
\definecolor{utgrey}{RGB}{97,82,88}
\definecolor{utlightgrey}{RGB}{173,175,175}

\definecolor{utsalmon}{RGB}{235,202,184}
\definecolor{utolive}{RGB}{136,123,27}
\definecolor{utred}{RGB}{198,12,48}
\definecolor{utwinered}{RGB}{130,36,51}
\definecolor{utdarkgrey}{RGB}{81,60,64}

	\title{A better lower bound for Lower-Left Anchored Rectangle Packing} 
	\author{
		Ruben~Hoeksma
		\thanks{University of Twente, Department of Applied Mathematics, The Netherlands. 
			\texttt{r.p.hoeksma@utwente.nl}, \texttt{m.t.maat@student.utwente.nl}
		}
		\and 
		Matthew~Maat\footnotemark[1]
	}

\begin{document}
		
	\maketitle
		
	\begin{abstract}
		Given any set of points~$S$ in the unit square that contains the
		origin, does a set of axis aligned rectangles, one for each point in~$S$, exist, such that each of them
		has a point in~$S$ as its lower-left corner, they are pairwise interior
		disjoint, and the total area that they cover is at least~$1/2$? This 
		question is also known as Freedman's conjecture (conjecturing that such 
		a set of rectangles does exist) and has been open since Allen Freedman 
		posed it in 1969. In this paper, we improve the best known lower bound 
		on the total area that can be covered from~$0.09121$ to~$0.1039$. Although 
		this step is small, we introduce new insights that push the limits of this 
		analysis. 
		
		Our lower bound uses a greedy algorithm with a particular order of the 
		points in~$S$. Therefore, it also implies that this greedy algorithm 
		achieves an approximation ratio of~$0.1039$. 
		We complement the result with an upper bound of~$3/4$ on the approximation 
		ratio for a natural class of greedy algorithms that includes the one that 
		achieves the lower bound.
	\end{abstract}

	\section{Introduction} 
		Consider the following packing problem that is best
		described by a two-player game between Alice and Bob. Alice is
		given the unit square~$U=[0,1]^2$ and may select any set of points~$S\subset
		U$ that includes the origin. After she selects this set, 
		Bob selects a set of axis-parallel rectangles 
		contained in the unit square,
		such that they are interior non-overlapping; none of them contains a
		point in~$S$ in their interior; and each has a point of~$S$ as their lower-left
		corner. We refer to such rectangles as \emph{lower-left-anchored axis-aligned 
		empty rectangles}.
		Bob's goal is to maximize the total area covered by his set of
		rectangles and Alice's goal is to minimize that same area.
		
		In 1969, Allen Freedman~\cite{Freedman69} conjectured that no matter what set of points Alice 
		chooses, Bob can always find a set of rectangles that satisfies the conditions 
		and covers an area of at least half of the unit square. This conjecture remains open 
		and it was even only in 2012 that Dumitrescu 
		and Tóth \cite{Dumitrescu2015,Dumitrescu2012} published the first result that shows that 
		Bob can always cover any 
		constant size area at all. Since then, the question of packing \emph{anchored 
		rectangles} in the unit-square has attracted more attention, yet, so far, no one 
		has improved the lower bound on what Bob can achieve. 
		
		In this paper, we improve the currently best lower bound by more than~$10\%$. 
		Our analysis pushes the boundary of what is possible with an analysis based 
		on the \tp{} algorithm that was suggested by Dumitrescu and 
		Tóth~\cite{Dumitrescu2015}. We make
		several key improvements, which allow for the better analysis. Our new lower
		bound is larger than~$0.1039$, which we believe pushes this analysis to its
		limits. 
		
		The \tp{} algorithm was designed by Dumitrescu and 
		Tóth in such a way that the \greedy{} algorithm that treats the points in the 
		same order covers an area that is at least as large. Therefore, the lower bound 
		implies that the \greedy{} algorithm covers at least a~$0.1039$-fraction of 
		what Bob's optimal solution would cover. Thus implying that the \greedy{} algorithm
		has a worst-case approximation guarantee for the optimization problem that Bob faces 
		of at least~$0.1039$. In 
		Section~\ref{sec:approximation}, we show, for a natural set of orders in which 
		the \greedy{} algorithm can treat the points in~$S$, that the \greedy{} algorithm 
		cannot have a worst-case approximation guarantee that is better than~$3/4$. 

	\subsection{Related work}
		The exact history of Freedman's Conjecture is not clear. While we know that 
		Allen Freedman posed the question in 1969 \cite{Freedman69}, it remains ambiguous if he 
		conjectured positively or negatively about the answer. Since that time,
		the problem has resurfaced at least twice~\cite{IBM2004,Winkler2007} before 
		Dumitrescu and Tóth booked the first real progress on it, 
		showing a lower bound of~$0.09121$~\cite{Dumitrescu2015}. 
		They also show that for any set $S$ an ordering exists such that the \greedy{} 
		algorithm that considers the points in that ordering gives the optimal solution 
		and that there exist orderings for which the \greedy{} algorithm cannot obtain 
		any constant covered area for some sets $S$. 
		
		From an approximation point of view, not much is known for the problem. We do not know 
		if it is NP-hard to decide if Bob can cover at least a certain area for a given 
		set of points, nor is there any approximation algorithm known beyond the one 
		implied by the lower bound. For variants of the problem, however, much more is known.
		When, instead of arbitrary rectangles, Bob is restricted to using only squares, 
		Balas er al.~\cite{Balas2017a} show an approximation ratio of~$1/3$. They also 
		show approximation guarantees for the setting where Bob may anchor rectangles 
		at any corner, achieving a~$7/12-O(1/n)$-approximation for rectangles and 
		a~$9/47$-approximation for squares in that setting, as well as a QPTAS and a PTAS 
		for rectangles and squares, respectively. 
		Akitaya et al.~\cite{DBLP:conf/mfcs/AkitayaJST18} prove that the latter variant 
		for squares is in fact NP-hard. Moreover, they show that for any instance with finite~$S$,
		the union of all feasible anchored square packings (which they call the reach) covers at 
		least an area of~$1/2$.
		Biedl at al.~\cite{Biedl2020} give a exact algorithms for 
		rectangles or squares anchored at any corner when 
		the points in~$S$ lie on the boundary of~$U$. 
		Antoniadis et al.~\cite{Antoniadis2019} introduce center anchored rectangle packing 
		(CARP), 
		where the point has to anchor to the center of a rectangle and the generalized 
		$(\alpha,\beta)$-anchored rectangle packing, where the point has to anchor 
		to the rectangle at the relative~$(\alpha,\beta)$-point. They show that CARP is NP-hard 
		and they show a PTAS for $(\alpha,\beta)$-anchored rectangle~packing. 
		
		From a more general point of view, Bob's Problem is among many different 
		problems that consider packing axis-aligned and
		interior-disjoint rectangles in a rectangular container.
		Some of the better-known problems in this category are 2D-knapsack and strip
		packing~\cite{DBLP:conf/focs/AdamaszekW13, DBLP:conf/soda/BansalK14}, 
		and the problem of finding a maximum-area independent set of given
		rectangles~\cite{DBLP:conf/soda/AdamaszekW15,
			DBLP:journals/ijcga/BeregDJ10, DBLP:journals/algorithmica/BeregDJ10}. 
		Rad\'{o} and Rado also formulated a whole range of similar 
		problems~\cite{Rado1,Rado2,Rado3,Rado0}.
	\section{Preliminaries} 
		Consider a set~$R$ of interior-disjoint axis-aligned rectangles in
		the unit square~$U=[0,1]^2$ and a set~$S\subset U$ that contains the origin. 
		For each point~$p\in S$ let~$x(p)$ and~$y(p)$ denote~$p$'s~$x$- and~$y$-coordinate, respectively.
		A rectangle is \emph{lower-left anchored} at a
		point~$p$ if~$p$ is the lower-left corner of that rectangle.
		We say that~$R$ is a \emph{lower-left anchored rectangle packing} (LLARP) of~$S$ 
		if each rectangle
		in~$R$ is an anchored rectangle, none of the rectangles contains any point in~$S$ 
		in its interior, and there is one anchored rectangle\footnote{Here, we consider a 
		single point or a line as a degenerate rectangle with zero area.} for each point.
	
		For any shape or set of shapes,~$P$, we denote by~$\area(P)$ the total area of 
		that shape or the total area of the union of the set of shapes. We can now define 
		Bob's Problem as follows.
		
		\begin{definition}[Bob's Problem]
			Given a set of~$n$ points~$S\subset U$ that contains the origin, 
			find a lower-left anchored rectangle packing~$R$ that maximizes~$\area(R)$.
		\end{definition}	
	
		Next we define two related algorithms, the \greedy{} algorithm and the \tp{} algorithm.
		Both algorithms can be defined for any order of the points in~$S$. For 
		this paper, we consider both only for an  
		order~$\succ$, such that~$p\succ q$ if~$x(p)+y(p)>x(q)+y(q)$. That is, 
		points are ordered increasingly according to the sum of their coordinates 
		and ties are broken arbitrarily but consistently.
		
		The \greedy{} algorithm finds an anchored rectangle packing as follows. 
		Iterate over all points in~$S$ in order of~$\succ$. Initiate~$R=\emptyset$. Consider the 
		current point and find a maximum area rectangle anchored at that point, 
		that is interior-disjoint with all rectangles in~$R$ and all points 
		in~$S$. Add this rectangle to~$R$. After the last iteration, return~$R$.
		
		The \tp{} algorithm first partitions the unit square into tiles, one for each
		point in~$S$ and then returns a set of rectangles with one largest-area
		rectangle within each tile. From the order~$\succ$, we obtain the tiling as
		follows. Iterate over all points in~$S$ in order of~$\succ$. From the current
		point cast two axis-aligned rays, one upwards and one to the right until they
		hit either the boundary of the unit square or one of the rays of a previously
		considered point. See also Figure~\ref{fig:tiling}. 
		Finally, from each tile, add a largest-area rectangle to~$R$. Note that any 
		rectangle with maximum area in a tile corresponding to a point~$p$ has~$p$ 
		as its lower-left~corner.
		
		The tiling of the \tp{} algorithm results in tiles that are simple rectilinear 
		polygons. Each concave corner of a tile corresponds to another point in~$S$.
		Intuitively, the
		tiling is such that the tile of a point~$p\in S$ contains all the area of the
		unit square that can never be covered by a rectangle anchored in a point
		$p'\in S$ such that~$p'\succ p$. From this intuition it is immediately clear
		that the \greedy{} algorithm always covers at least as much as what the \tp{}
		algorithm does. This is formalized by Lemma~\ref{lem:tilepacking}.

		\begin{lemma}[\cite{Dumitrescu2015}, Lemma~2.1]\label{lem:tilepacking}
			For each point~$p\in S$, the \greedy{} algorithm adds a rectangle anchored 
			at~$p$ with area at least as large as the rectangle anchored at~$p$ that 
			the \tp{} algorithm adds.
		\end{lemma}

	\section{An improved lower bound for Bob's Problem} 

		In this section, we prove an improved lower bound of~$0.10390$ on the 
		area that Bob can cover for any set of points~$S$ that includes the origin
		by using the \greedy{} algorithm. 
		Our proof follows the one by
		Dumitrescu and Tóth~\cite{Dumitrescu2015}. Like they do, we prove that the
		\tp{} algorithm finds a solution that covers at least a constant
		fraction of the unit square. The advantage of considering the \tp{}
		algorithm over the \greedy{} algorithm is that, while the \greedy{} algorithm covers
		at least as much as the \tp{} algorithm, the \tp{} algorithm can
		be easily analyzed locally, per tile. For each tile, we can simply look at its
		largest contained rectangle. When this rectangle covers less than a~$1/\beta$
		fraction of the tile we call the tile a~\emph{$\beta$-tile}. We first bound the area
		of a~$\beta$-tile in terms of the area of two specially defined
		parallelograms (Lemma~\ref{lem:lb1}). Then, we bound the total area of 
		the parallelograms through a charging scheme (Lemmata~\ref{lem:lb2} and~\ref{lem:lb3}). 
		This implies a bound on the total area of the~$\beta$-tiles. 
		The remaining area of the unit square must therefore consist of tiles which have 
		a largest contained rectangle that has an area of at least a~$1/\beta$ fraction 
		of the area of the tile. Naturally, the set of~$\beta'$-tiles includes the set 
		of~$\beta$-tiles for~$\beta'\le\beta$ and our bound on the total area of~$\beta$-tiles 
		is decreasing in~$\beta$. We optimize the analysis by integrating over all 
		possible~$\beta$ that contribute positively to the area. 

		Our analysis improves on that of Dumitrescu and Tóth~\cite{Dumitrescu2012}
		in three ways. We introduce a variable~$\alpha$ that varies the definition 
		of the tips of~$\beta$-tiles (to be defined later), we bound the area of 
		such tiles by parallelograms instead of trapezoids, and we bound all~$\beta$-tiles 
		at once with the use of the inequality of arithmetic and geometric means.
		
		\begin{figure} 
			\centering 
			\begin{tikzpicture}[scale=5]
				\tikzstyle{point} = [fill=utblue, inner sep=2pt, draw=black, circle]
				\coordinate (bottomright) at (1,0);
				\coordinate (topright) at (1,1);
				\coordinate (topleft) at (0,1);
				\draw[thick] (0,0) -- (bottomright) -- (topright) -- (topleft) -- cycle;
				\node[point, label ={[label distance=-5pt]above right:$(0,0)$}] (O) at (0,0) {};
				
				\node[point] (p1) at (0.6,0.9) {};
				\draw[very thick] (p1.east) -- (p1 -| bottomright);
				\draw[very thick] (p1.north) -- (p1 |- topright);
				\node[point] (p2) at (0.65,0.65) {};
				\draw[very thick] (p2.east) -- (p2 -| bottomright);
				\draw[very thick] (p2.north) -- (p2 |- p1);
				\node[point] (p3) at (0.75,0.3) {};
				\draw[very thick] (p3.east) -- (p3 -| bottomright);
				\draw[very thick] (p3.north) -- (p3 |- p2);
				\node[point] (p4) at (0.3,0.7) {};
				\draw[very thick] (p4.east) -- (p4 -| p2);
				\draw[very thick] (p4.north) -- (p4 |- topright);
				\node[point] (p5) at (0.1,0.85) {};
				\draw[very thick] (p5.east) -- (p5 -| p4);
				\draw[very thick] (p5.north) -- (p5 |- topright);
				\node[point] (p6) at (0.55,0.1) {};
				\draw[very thick] (p6.east) -- (p6 -| topright);
				\draw[very thick] (p6.north) -- (p6 |- p4);
				\node[point] (p7) at (0.2,0.35) {};
				\draw[very thick] (p7.east) -- (p7 -| p6);
				\draw[very thick] (p7.north) -- (p7 |- p5);
			\end{tikzpicture}
			\caption{The tiling of a point set in the unit square that the \tp{} algorithm uses.} 
			\label{fig:tiling}
		\end{figure}
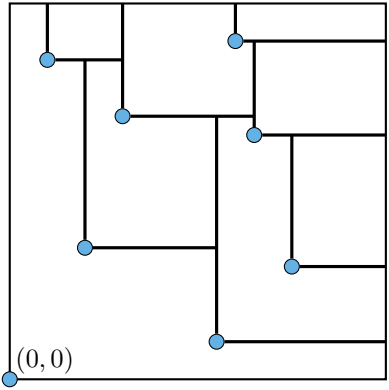

		We name several parts of the tiles that the \tp{} algorithm uses.
		We denote the tiles by~$t_1,\ldots,t_n$, one for each point in~$S$. 
		A tile~$t_i$ is called \emph{a~$\beta$-tile} if the largest possible 
		anchored rectangle contained in~$t_i$ has an area of less than~$\frac{1}{\beta}\area(t_i)$
		for some fixed constants~$\alpha$ and~$\beta\geq 3+2\alpha$.
		Let the \emph{right tip} of~$t_i$ be the smallest sub-polygon of~$t_i$
		that contains all points of~$t_i$ to the right of a vertical line 
		through a concave corner of~$t_i$ 
		and has an area of at least~$\frac{\alpha}{\beta}\area(t_i)$. Similarly, let the
		\emph{upper tip} of~$t_i$ be the smallest sub-polygon of~$t_i$
		that contains all points of~$t_i$ above a horizontal line through a convex corner of~$t_i$
		and has an area of at least~$\frac{\alpha}{\beta}\area(t_i)$.
		We refer to the polygon that consists of all points in~$t_i$ that are not part of its upper 
		and right tip as the \emph{main body},~$t_i'$.  
		By~$a_i$ and~$b_i$ we denote the bottom and left edge of~$t_i$, respectively. 
		Similarly we denote by~$a_i'$ the bottom edge of the main body~$t_i'$ and 
		by~$b_i'$ the left edge of~$t_i'$.
		Figure~\ref{fig:lowerbound1a} depicts a~$\beta$-tile and the mentioned parts.
		
		Next, we define a number of support shapes that let us bound the area of a~$\beta$-tile. 
		Let~$\Delta_i$ be the isosceles right triangle with~$a_i$ as
		its top edge and its~$90^\circ$ angle at the bottom right point
		of~$t_i$. Similarly, let~$\Gamma_i$ be the isosceles right triangle with~$b_i$
		as its right edge and its~$90^\circ$ angle at the top left point of~$t_i$.
		Let~$\lambda$ be a constant such that $\lambda\in(0,\alpha)$ and let~$A_i$ 
		be the parallelogram with base~$a_i'$ and height~$\lambda|a_i'|$ and two 
		sides parallel to the diagonal edge of~$\Delta_i$. Similarly, let~$B_i$ be the
		parallelogram with base~$b_i'$, two sides parallel to the diagonal edge of~$\Gamma_i$,
		and height equal to~$\lambda|b_i'|$. 
		Figure~\ref{fig:lowerbound1a} also depicts these support shapes.
		
		Finally, we state a reformulation of a lemma from~\cite{Dumitrescu2015} 
		that we need in the proof of~Lemma~\ref{lem:lb1}.

		\begin{lemma}[Lemma~3.1 in \cite{Dumitrescu2015}]\label{lem:DT3.1} 
			Let~$t$ be a $\beta$-tile with height~$h$ and width~$w$, then 
			\[
				\area(t)<\frac{\beta}{e^{\beta-1}}hw\,.
			\]
		\end{lemma}

		\begin{figure} 
			\centering %
			\includegraphics[width=10cm]{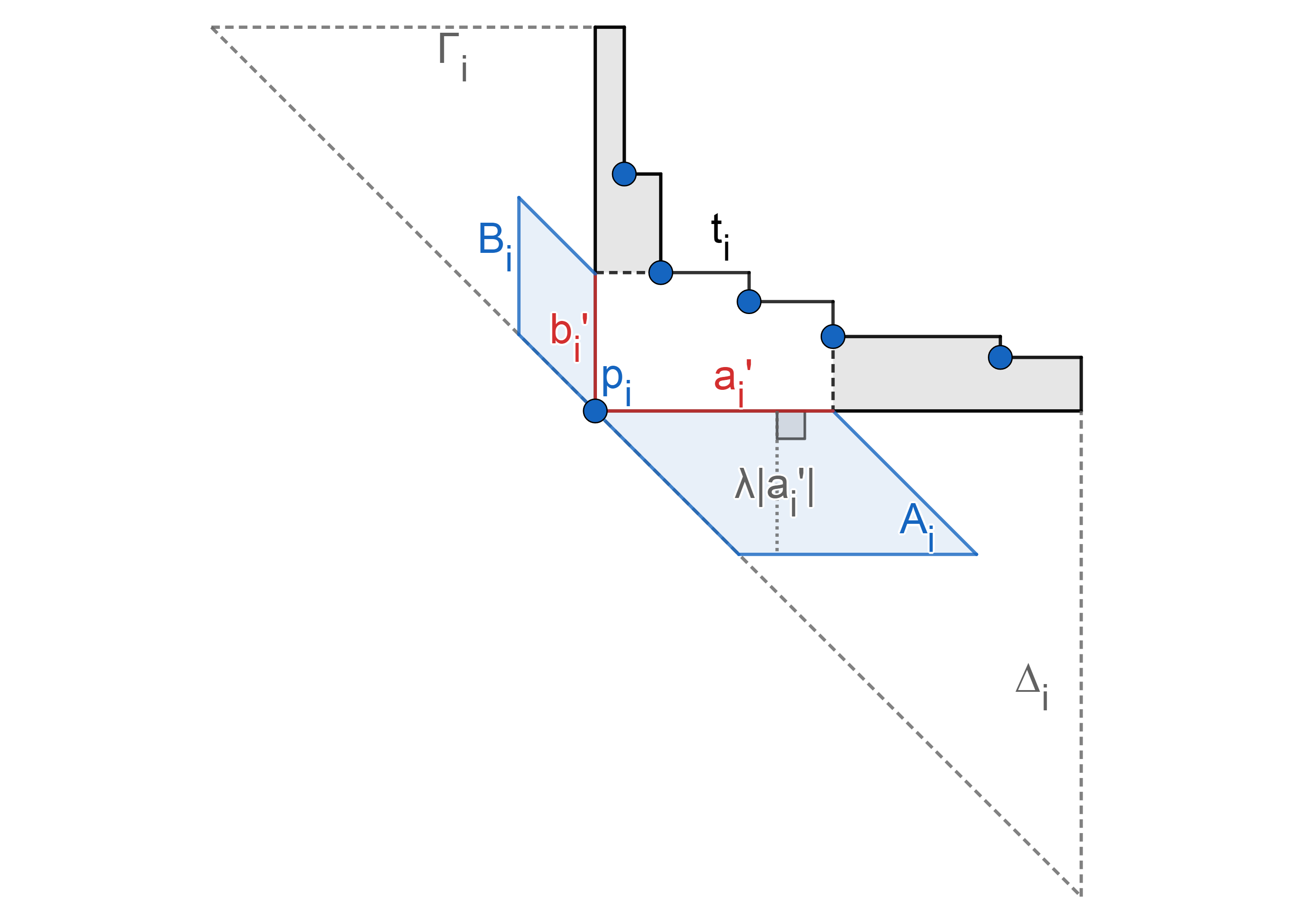}%
			\caption{A~$\beta$-tile~$t_i$, with its right and upper tip in grey.
				Triangles~$\Delta_i$ and~$\Gamma_i$ are marked with dashed lines,~$a_i'$
				and~$b_i'$ are depicted in red, and parallelograms~$A_i$ and~$B_i$ are depicted in blue.}%
			\label{fig:lowerbound1a}
		\end{figure}
	
		\begin{lemma} \label{lem:lb1} 
			The area of the~$\beta$-tile~$t_i$ is at most
			$\frac{\beta}{2\lambda e^{\beta-3-2\alpha}}\left(\area(A_i)+\area(B_i)\right)$. 
		\end{lemma}
	
		\begin{proof} 
			Note first that the area of the right tip is less than~$\frac{1+\alpha}{\beta}\area(t_i)$. 
			If its area was larger, we could move
			its left boundary to the next concave corner (one `step' on the staircase) 
			to the right, and,
			since the difference is a rectangle within~$t_i$, the area would be reduced
			by less than~$\frac{1}{\beta}\area(t_i)$. Therefore, the remaining area would be 
			greater than~$\frac{\alpha}{\beta}\area(t_i)$, which contradicts that the tip is 
			the smallest such polygon.
			Similarly, the area of the upper tip is less
			than~$\frac{1+\alpha}{\beta}\area(t_i)$. Therefore, the area of the main 
			body,~$t_i'$, is at least~$\frac{\beta-2-2\alpha}{\beta}\area(t_i)$, which is at
			least~$\frac{1}{\beta}\area(t_i)$, as~$\beta\geq 3+2\alpha$. This implies
			that the right and upper tip do not overlap, since if they did, the main body would
			be a rectangle and cannot have an area of at least~$\frac{1}{\beta}\area(t_i)$. 
			Furthermore, the area of the largest rectangle inside~$t_i'$ has area less than 
			\[
				\frac{1}{\beta}\area(t_i) \le \frac1\beta\frac{\beta}{\beta-2-2\alpha}\area(t_i')
				=\frac{1}{\beta-2-2\alpha}\area(t_i')\,. 
			\]
			Thus $t_i'$ is a $\frac{1}{\beta-2-2\alpha}$-tile and from Lemma~\ref{lem:DT3.1} 
			we get that 
			$\area(t_i')<\frac{\beta-2-2\alpha}{e^{\beta-3-2\alpha}}|a_i'||b_i'|$.

			Since,~$\area(A_i)=\lambda|a_i'|^2$,~$\area(B_i)=\lambda|b_i'|^2$ and~$\area(t_i')>\frac{\beta-2-2\alpha}{\beta}\area(t_i)$, we get 
			\begin{align*} 
				\area(t_i)&\le\frac{\beta}{\beta-2-2\alpha}\area(t_i')\\
				&<\frac{\beta}{\beta-2-2\alpha}\frac{\beta-2-2\alpha}{e^{\beta-3-2\alpha}}|a_i'||b_i'|\\
				&\leq \frac{\beta}{e^{\beta-3-2\alpha}}\frac{1}{2}(|a_i'|^2+|b_i'|^2)\\
				&=\frac{\beta}{2\lambda e^{\beta-3-2\alpha}}\left(\area(A_i)+\area(B_i)\right) \,, 
			\end{align*} 
			where the last inequality is due to the arithmetic mean-geometric mean inequality.
		 \end{proof}

		Next, we bound the total area of all parallelograms~$A_i$ of
		$\beta$-tiles. To do so, we define a directed graph~$G$, where the nodes
		correspond to the~$A_i$ parallelograms. If a parallelogram~$A_i$
		intersects some other parallelograms, then~$A_i$ gets exactly one outgoing edge
		to the node of the parallelogram~$A_j$ that intersects~$A_i$ and of which the
		corresponding line segment~$a_j'$ is the highest below~$a_i'$. 
		Here, we assume that all points have distinct $y$-coordinates. If this is not 
		the case, we introduce an ordering of the points $\succ_y$ such that $p\succ_y q$ 
		if $y(p)>y(q)$ and ties are broken arbitrarily, but consistently. Then,~$A_i$ 
		gets exactly one outgoing edge
		to the node of the parallelogram~$A_j$ that intersects~$A_i$ and is first 
		according to~$\succ_y$.
		Parallelograms that do not have an outgoing edge are \emph{at level 1}. All
		other parallelograms are \emph{at level~$k+1$} if they have an outgoing edge to a level~$k$ 
		parallelogram. This constructs an acyclic graph. 
		
		We use a charging scheme for the parallelograms, where the area of each 
		parallelogram~$A_i$ at level~$2$ and higher is charged to the unique 
		parallelogram at level~$1$ that~$A_i$ has a directed path to. Before that, 
		we first derive an upper bound on the total area of the level~$1$ parallelograms:
	
		\begin{lemma} \label{lem:lb2} 
			The total area of all~$A_i$-parallelograms at level~$1$ is at
			most~$\frac{2(1+\alpha)^2+\lambda(2+\alpha)}{2(1+\alpha)^2}$.
		\end{lemma}
		\begin{proof} 
			We show that all level~$1$ $A_i$-parallelograms lie in the hexagon with 
			corner points 
			\[\{(0,0),(0,1),(1,1),(1,0), 
			(\frac{\lambda+1}{1+\alpha},\frac{-\lambda}{1+\alpha}),
			(\frac{\lambda}{1+\alpha},\frac{-\lambda}{1+\alpha})\}\]
			(see right half of Figure~\ref{fig:lowerbound34}). 
			This hexagon consists of the unit square and a
			trapezoid with height~$\frac{\lambda}{1+\alpha}$ and base lengths~$1$
			and~$\frac{1}{1+\alpha}$, therefore it has an area of
			\[
				1+\frac{1}{2}\frac{\lambda}{1+\alpha}\left(1+\frac{1}{1+\alpha}\right)
				=\frac{2(1+\alpha)^2+\lambda(2+\alpha)}{2(1+\alpha)^2}\,.
			\] 
			Clearly, no parallelogram crosses the top or left boundary of~$U$.
			The largest rectangle in~$t_i'$ with~$a_i'$ as its bottom edge has
			area at most~$\frac{1}{\beta}\area(t_i)$. Moreover, since~$t_i'$ is higher 
			than the right tip of~$t_i$, this rectangle also has a larger height than the
			right tip, which has area at least~$\frac{\alpha}{\beta}\area(t_i)$ (see
			Figure~\ref{fig:lowerbound1a}). The area of the right tip of~$t_i$ is 
			at most its height times its width,~$|a_i|-|a_i'|$. 
			Therefore,~$\alpha|a_i'|<|a_i|-|a_i'|$ and~$(1+\alpha)|a_i'|<|a_i|$. 
			Since~$\lambda<\alpha$ and since the bottom-right corner of~$A_i$ is 
			exactly~$\lambda|a_i'|$ 
			to the right of line segment~$a_i'$, the bottom-right corner of~$A_i$ has 
			$y$-coordinate at most 
			\[
				|a_i'| + \lambda|a_i'| < (1+\alpha)|a_i' < |a_i|
			\]
			larger than the $y$-coordinate of $p_i$. Therefore,~$A_i$ lies completely 
			inside~$\Delta_i$ and no parallelogram crosses the right boundary of~$U$. 
			
			Since the bottom-right corner of~$A_i$ is~$\lambda|a_i'|$ below~$a_i'$, 
			the line through the bottom-right corner of~$A_i$ and the point at the 
			right end of~$a_i$ has a
			slope of at most~$\frac{\lambda}{\alpha-\lambda}$ (see the left half of 
			Figure~\ref{fig:lowerbound34}). Since the point at the right end of~$a_i$ 
			lies inside~$U$, the parallelograms must therefore lie above the line
			$y=\frac{\lambda}{\alpha-\lambda}(x-1)$.  
			
			No parallelogram crosses
			the line~$y=-x$, as the parallelograms have an angle of~$45^\circ$. 
			Since~$|a_i|\leq 1$,
			we have~$|a_i'|<\frac{1}{1+\alpha}$, so the height of all parallelograms is
			less than~$\frac{\lambda}{1+\alpha}$. Therefore, all parallelograms are above
			the line~$y=-\frac{\lambda}{1+\alpha}$ and they all lie in the hexagon. 
			Since, by construction, no two level~$1$ parallelograms overlap their 
			total area is bounded by~$\frac{2(1+\alpha)^2+\lambda(2+\alpha)}{2(1+\alpha)^2}$.
		\end{proof} 
	
		\begin{figure}
			\centering 
			\includegraphics[width=\linewidth]{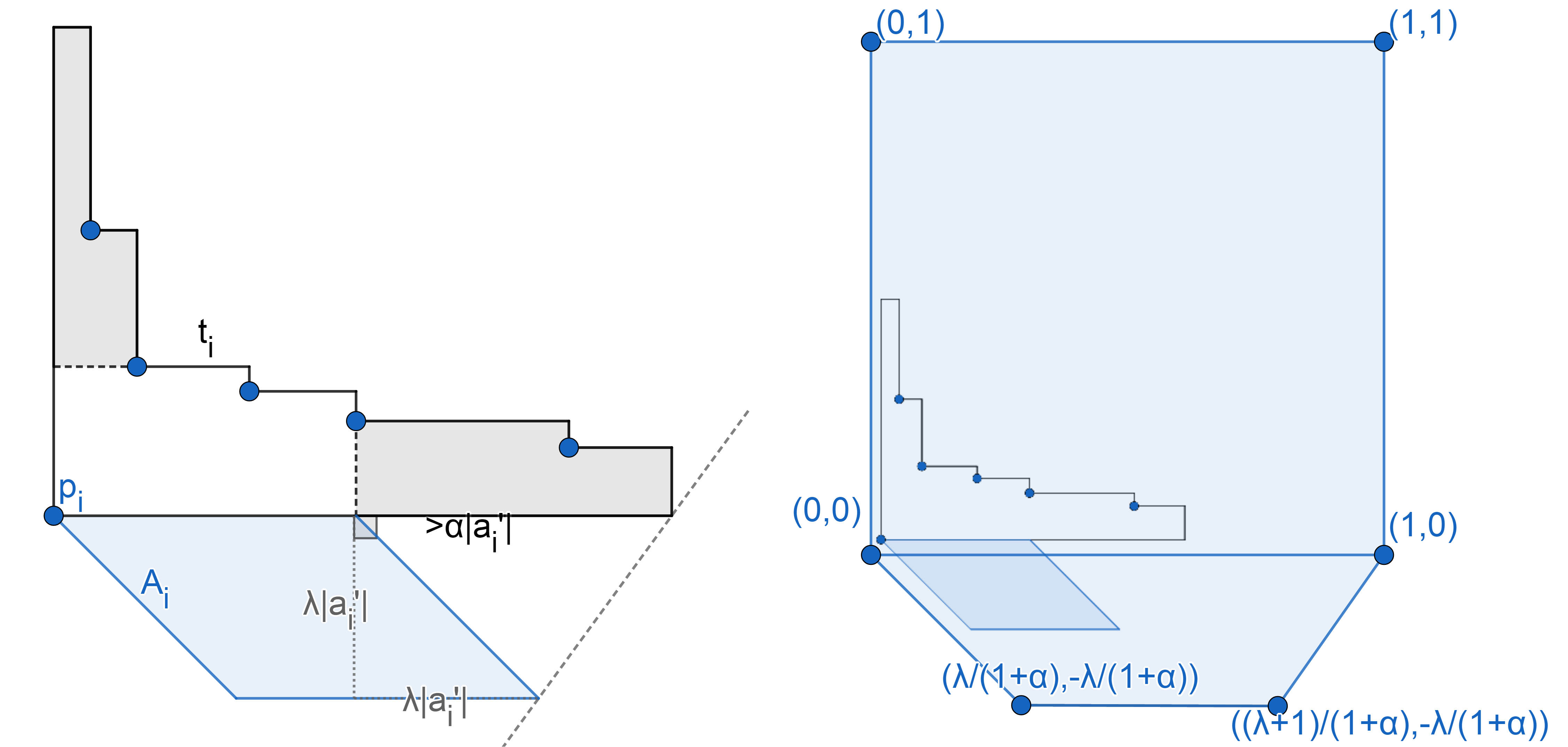} 
			\caption{Left: the slope of the line through the bottom right corner of~$A_i$ 
				and the right end of~$a_i$ is at most~$\frac{\lambda}{\alpha-\lambda}$. 
				Right: all parallelograms lie inside the blue hexagon.} \label{fig:lowerbound34} 
		\end{figure}
	
		Now we derive an upper bound on the area of the parallelograms of level~$2$ and
		higher. We do this by charging their area to the level~$1$ parallelograms. 
		In the proof, we need the following lemma from \cite{Dumitrescu2015}.
		
		\begin{lemma}[Lemma~3.3 in \cite{Dumitrescu2015}]\label{lem:DT3.3}
			For every~$i\in\{1,\ldots,n\}$, the interior of $\Delta_i$ (resp., $\Gamma_i$) is disjoint
			from $S$.
		\end{lemma}
	
		\begin{lemma} \label{lem:lb3} 
			For every parallelogram~$A_j$ at level 1, the total
			area of all parallelograms~$A_i$, with~$i\neq j$, 
			for which there is a directed paths in~$G$ to~$A_j$,  
			is at most~$\frac{1}{2(\alpha-\lambda)}\area(A_j)$.
		\end{lemma} 
	
		\begin{figure} 
			\centering
			\includegraphics[width=\linewidth]{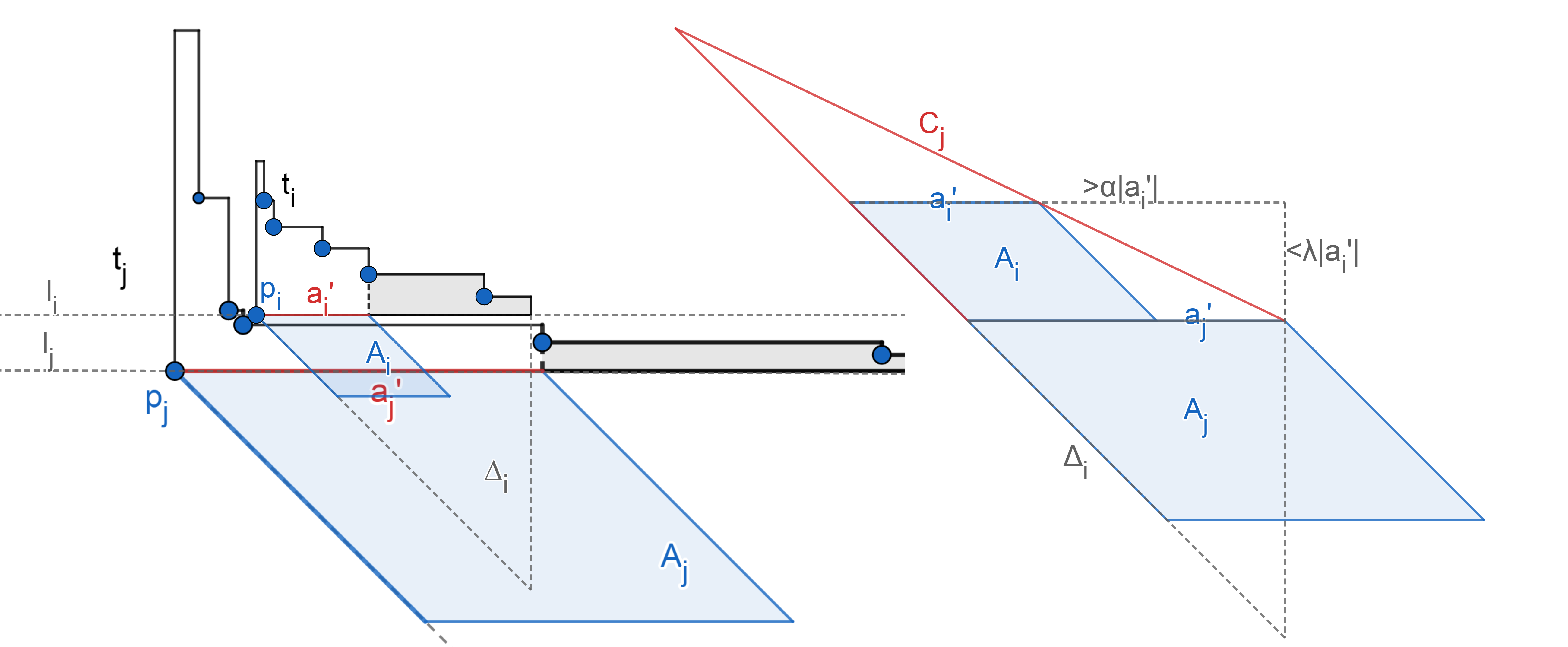}
			\caption{Left: parallelograms~$A_i$ and~$A_j$ intersect and~$a_i'$ lies 
				above~$a_j'$. Right: parallelograms~$A_i$ and~$A_j$ intersect 
				and~$a_i'$ lies above~$a_j'$.}
			\label{fig:lowerbound56} 
		\end{figure}
	
		\begin{proof} 
			For all~$i$, denote by~$l_i$ the line through~$a_i$. 
			We consider a parallelogram~$A_i$ that intersects a
			parallelogram~$A_j$ where~$a_i'$ lies above~$a_j'$. That is, there is an 
			edge from~$A_i$ to~$A_j$ in the graph~$G$.
			Consider the intersection of~$\Delta_i$ with~$l_j$ (see Figure~\ref{fig:lowerbound56}).
			If the intersection of~$A_i$ 
			with~$l_j$ contains any point of~$U$ left of~$a_j'$, the point~$p_j$ 
			lies inside~$A_i$ and hence inside~$\Delta_i$.
			If it contains any point to the right of~$a_j'$, then the top left corner of 
			the right tip of~$t_j$, which is a point in~$S$, 
			lies inside~$\Delta_i$. 
			By Lemma~\ref{lem:DT3.3}, we know that the intersection of~$S$ 
			and the interior of~$\Delta_i$ is empty. Therefore, the
			intersection of~$\Delta_i$ with~$l_j$ must be a subset of~$a_j'$.
			
			Define
			the triangle~$C_j$ as the triangle that is bounded by~$a_j'$, by the line with
			slope~$-1$ through~$p_j$, and the line with slope~$-\frac{\lambda}{\alpha}$
			through the point at the right end of~$a_j'$. 
			We show that~$A_i$ lies completely inside~$C_j\cup A_j$.
			See the right half of Figure~\ref{fig:lowerbound56}.
			Since the intersection of $\Delta_i$ with~$l_j$ is a subset of~$a_j$, 
			which is one side of~$A_J$, 
			the parallelogram~$A_i$ is not larger than~$A_j$ and the part of~$A_i$ that 
			lies below~$l_j$ is contained in~$A_j$. Thus,
			it suffices to 
			show that~$A_i$ lies to the right of 
			the line with slope~$-1$ through~$p_j$ and under 
			the line with slope~$-\frac{\lambda}{\alpha}$
			through the point at the right end of~$a_j'$.
			
			Again, since the intersection of~$\Delta_i$
			with~$l_j$ is a subset of~$a_j'$, the parallelogram~$A_i$ must lie
			to the right of the line with slope~$-1$ through~$p_j$. 
			
			Since~$(1+\alpha)|a_i'|<|a_i|$, the right edge of~$\Delta_i$ (and also the
			right end point of~$a_j'$) lies at least~$\alpha|a_i'|$ to the right of the 
			right end point of~$a_i'$. Since~$A_i$ has height~$\lambda|a_i'|$, the 
			segment~$a_i'$ lies at most~$\lambda|a_i'|$ above~$a_j'$. 
			Thus, the right end of~$a_i'$ lies below the line with slope~$-\frac{\lambda}{\alpha}$
			through the point at the right end of~$a_j'$. 
			Thus,
			\begin{equation}
				\tag{i}\label{statement:i}\text{if~$A_i$ intersects~$A_j$ and~$a_i'$ lies 
				above~$a_j'$, then~$A_i$ lies inside~$C_j\cup A_j$.}
			\end{equation}
			
			Now, let~$\Xi_i$ be the strip that is bounded by the two lines tangent to the
			left and right edge of~$A_i$ (see Figure~\ref{fig:lowerbound7}). We prove by
			induction that the strips of parallelograms at
					the same level that have a directed path to the same parallelogram
					$A_j$ are interior-disjoint.
			
			As induction basis, we see that clearly the parallelograms at level~$1$ 
			have interior-disjoint~$\Xi_i$, as there is only one parallelogram.  
			
			Now, suppose as induction hypothesis that the strips of the parallelograms at
			level~$k$ that have a directed path to the same level~$1$ parallelogram are
			all interior-disjoint. Since $C_j$ lies in the strip~$\Xi_j$, we have from
			\eqref{statement:i} that~$\Xi_i\subseteq\Xi_j$ if there is an edge from~$A_i$
			to~$A_j$ in~$G$. The parallelograms at level~$k+1$ each have
			an outgoing edge to a level~$k$ parallelogram, so it follows from the
			induction hypothesis that the strips~$\Xi_i$ are interior-disjoint if they
			have outgoing edges to different level~$k$ parallelograms. 
			
			Now consider two level~$k+1$ parallelograms~$A_{i_1}$ and~$A_{i_2}$ with outgoing
			edges to the same parallelogram~$A_j$. Suppose that~$\Xi_{i_1}$ and~$\Xi_{i_2}$ 
			are not interior-disjoint. 
			
			If~$A_{i_1}$ and~$A_{i_2}$ intersect, w.l.o.g.,~$a_{i_1}$ lies higher than~$a_{i_2}$, 
			but then~$A_{i_1}$ has an outgoing edge to~$A_{i_2}$ and not to~$A_j$, which
			is a contradiction. See Figure~\ref{fig:lowerbound7}.
			
			If~$A_{i_1}$ and~$A_{i_2}$ do not
			intersect, then one parallelogram must lie completely above the other.		
			W.l.o.g., let~$A_{i_1}$ lie above~$A_{i_2}$. However, then~$A_{i_1}$ cannot
			intersect~$A_j$, as its bottom edge lies above~$a_{i_2}'$, which lies above 
			the top edge of~$A_j$ since~$A_{i_2}$ has an edge to~$A_j$. Again, a contradiction.
			
			Thus, all parallelograms at level~$k+1$ that have directed paths to the same level~$1$
			parallelogram have interior-disjoint~$\Xi_i$.

		\begin{figure}[tbp] 
			\centering
			\includegraphics[width=0.48\textwidth]{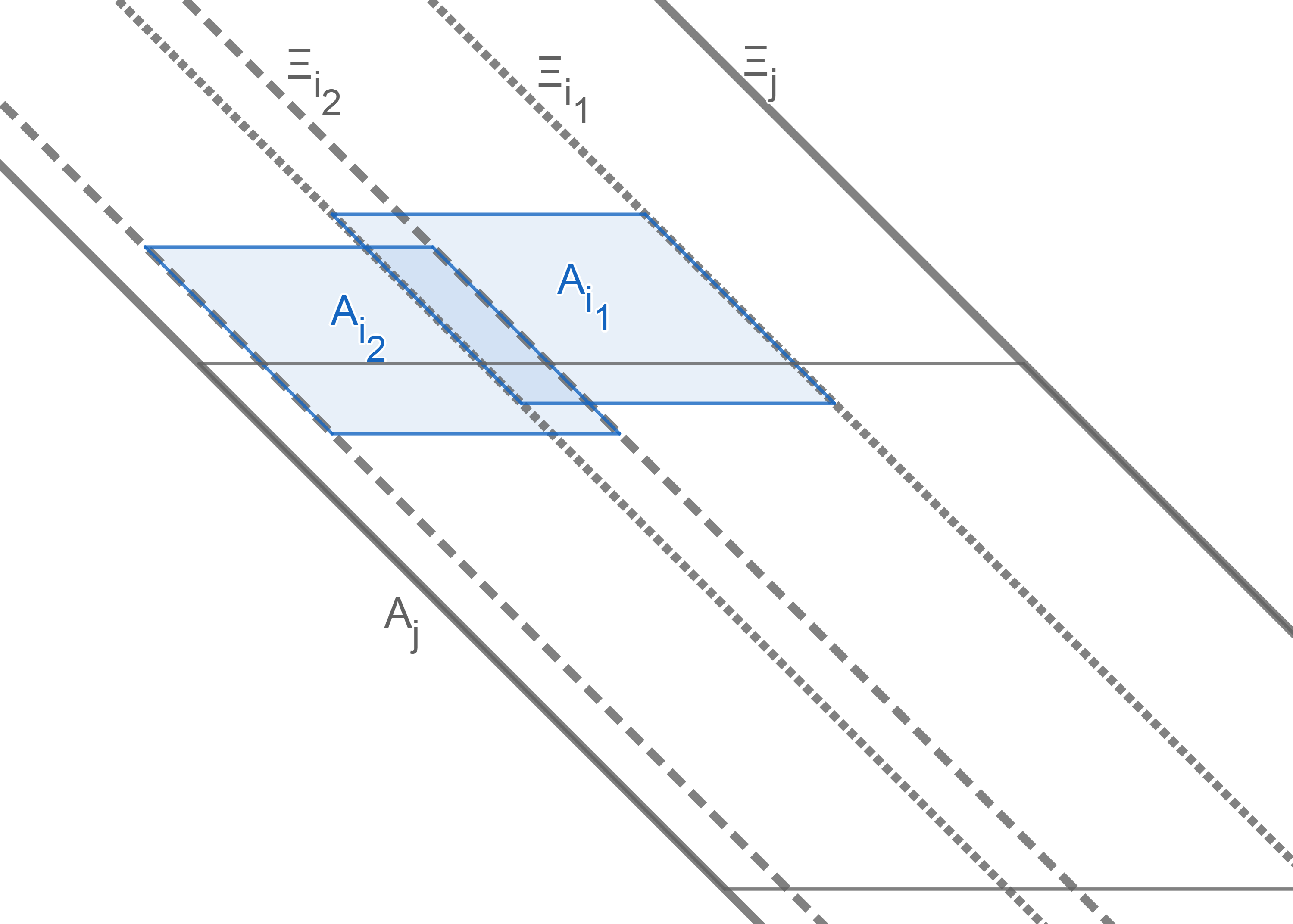}
			\caption{The strips~$\Xi_{i_1}$,~$\Xi_{i_2}$ and~$\Xi_j$. If the strips
				$\Xi_{i_1}$ and~$\Xi_{i_2}$ overlap, parallelograms~$A_{i_1}$ and~$A_{i_2}$
				overlap, and if~$a_{i_1}$ is above~$a_{i_2}$, there cannot be an edge from
				$A_{i_1}$ to~$A_j$.} \label{fig:lowerbound7} 
		\end{figure}
		
			Now, for a parallelogram~$A_j$ at level $1$, we apply a translation on all
			parallelograms~$A_i$ with a directed path to~$A_j$ in~$G$ as follows. 
			We translate all parallelograms per level increasing in level. When level~$k$ 
			is translated, translate all parallelograms one
			by one, parallel to the line~$y=-x$ upwards, and translate each parallelogram
			exactly far enough so that their interior does not overlap with any parallelogram at a
			lower level. Since they also do not intersect with parallelograms in their own level,
			no two parallelograms overlap at the end of the translation.
			
			We prove that after this transformation all parallelograms lie inside~$C_j$. 
			Again, we prove this by induction on the levels. 
			
			For the induction basis, no two parallelograms at level~$2$ can
			overlap after the translation. Now, since after the translation level~$2$ 
			parallelograms still 
			intersect~$A_j$ with their boundary, by~\eqref{statement:i}, they lie inside~$C_j$. 
			
			Now, suppose as induction hypothesis that all
			translated parallelograms fit inside~$C_j$ after level~$k$ is translated. Then, any
			parallelogram~$A_i$ at level~$k+1$ has an edge to a level~$k$ parallelogram,
			say to~$A_{i_0}$. Again by~\eqref{statement:i}, after the translation~$A_i$
			fits inside the translated version of~$C_{i_0}$. This triangle~$C_{i_0}$ fits 
			inside~$C_j$, since its bottom edge lies inside~$C_j$ and it is similar to~$C_j$.
			Therefore, the parallelograms at level~$k+1$ also fit inside~$C_j$.
			
			The base of~$C_j$ is~$a_j'$ and its height 
			is~$\frac{\lambda}{\alpha-\lambda}|a_j'|$. Therefore the area 
			of all parallelograms with a directed path to~$A_j$ is bounded by
			\[
				\frac{1}{2}\frac{\lambda}{\alpha-\lambda}|a_j'|^2
				=\frac{1}{2(\alpha-\lambda)}\area(A_j)\,.\qedhere
			\]
		\end{proof}
	
		Combining Lemmata \ref{lem:lb2} and \ref{lem:lb3} we get an upper bound on the
		area of the parallelograms~$A_i$. By symmetry, the same bound holds for the
		parallelograms~$B_i$. Combining these bounds with Lemma \ref{lem:lb1}, we get
		\begin{align*} 
			\sum_{t_i \text{ is a~$\beta$-tile}}\area(t_i) &\stackrel{\ref{lem:lb1}}{<} \sum_{i
			\text{ for } t_i \text{ a~$\beta$-tile}} \frac{\beta}{2\lambda
			e^{\beta-3-2\alpha}}\left(\area(A_i)+\area(B_i)\right)\\
			&\stackrel{\ref{lem:lb3}}{\leq} \frac{\beta}{2\lambda 
			e^{\beta-3-2\alpha}}\left(1+\frac{1}{2(\alpha-\lambda)}\right)\left(\sum_{A_i \text{level 1}} \area(A_i)+\sum_{B_i \text{level 1}} \area(B_i)\right)\\
			&\stackrel{\ref{lem:lb2}}{\leq} \frac{\beta}{2\lambda e^{\beta-3-2\alpha}}\frac{1+2\alpha-2\lambda}{2\alpha-2\lambda}\cdot 2\cdot\frac{2(1+\alpha)^2+\lambda(2+\alpha)}{2(1+\alpha)^2}\,.
		\end{align*} 
		Thus, we define 
		\begin{equation}
			F(\beta,\lambda,\alpha)=\frac{(1+2\alpha-2\lambda)(2(1+\alpha)^2+\lambda(2+\alpha))e^{3+2\alpha}}
			{4\lambda(\alpha-\lambda)(1+\alpha)^2 }\frac\beta{e^{\beta}}\,,\label{eq:area}
		\end{equation}
		which bounds the total area of~$\beta$-tiles from above. 
		At this point, we can bound the area that the \tp{} algorithm covers from below by
		realizing that all tiles that are not~$\beta$-tiles contain rectangles that
		cover at least~$1/\beta$ of their total area. Therefore, we can bound the total area that 
		the \tp{} algorithm covers by
		\begin{align*}
			\area(R) \ge \frac1\beta\left(1-F(\beta,\lambda,\alpha)\right)
			= 	\frac1\beta - \frac{(1+2\alpha-2\lambda)(2(1+\alpha)^2+\lambda(2+\alpha))e^{3+2\alpha}}
			{4\lambda(\alpha-\lambda)(1+\alpha)^2 }\frac1{e^{\beta}}\,.
		\end{align*}
		Optimization of the parameters through numerical methods\footnote{See Appendix~\ref{app:0.08}.} has provided us 
		with the parameter values~$(\beta,\lambda,\alpha) = (11.31,0.7696,0.4456)$. With which we obtain a 
		bound on the total area covered of~$0.0806$. This bound can be further improved 
		by considering~$\beta$ as a continuous variable and integrating over all values of $\beta$ 
		that contribute positively to the bound. This method is explained by 
		Dumitrescu and Tóth~\cite{Dumitrescu2015}. This provides a bound of 
		
		\begin{equation} 
			\area(R)\geq \frac1{\beta_0} - \frac{(1+2\alpha-2\lambda)(2(1+\alpha)^2+\lambda(2+\alpha))e^{3+2\alpha}}
			{4\lambda(\alpha-\lambda)(1+\alpha)^2 }
			\int_{\beta_0}^\infty 
			\frac{1}{\beta e^{\beta}}\ud\beta\,.\label{eq:integral}
		\end{equation}
		We have included a proof of this bound in Appendix~\ref{app:integral}.

		Again, by using numerical optimization methods\footnote{See Appendix~\ref{app:0.10}.} we found that the 
		parameter values~$(\beta_0,\lambda,\alpha)=(8.6142,0.44581,0.76975)$
		yield~$\area(R)\geq 0.10390$. 		
		
		\begin{theorem}\label{thm:Lowerbound}
			The \greedy{} algorithm yields an anchored rectangle packing 
			with an area of at least~$0.10390$.
		\end{theorem}

	\section{The GreedyPacking algorithm as an approximation algorithm}
	\label{sec:approximation}
		In this section, we treat the \greedy{} algorithm as an approximation algorithm
		for Bob's Problem. As such, we aim to bound the \emph{approximation ratio} of
		the algorithm, the worst-case ratio over all possible instances of the 
		area covered by a solution that the algorithm produces and the area 
		covered by an optimal solution. Since no optimal solution can ever cover more 
		than the whole unit square, Theorem~\ref{thm:Lowerbound} implies that the 
		approximation ratio of the \greedy{} algorithm is at least~$0.10390$. In this 
		section, we prove that the approximation ratio of the \greedy{} algorithm 
		for a large group of orderings is at most~$\frac{3}{4}$.
		
		Before we formally state the theorem for this section, we need some definitions. 
		We say that a function~$g:U\to \mathbb{R}$ is \emph{symmetric}
		if~$g((x,y))=g((y,x))$ for all~$(x,y)\in U$ and we say that~$g$ is
		\emph{strictly increasing} in~$x$, respectively~$y$, if~$g((x,y))>g((x',y))$ if~$x>x'$, 
		respectively if~$g((x,y))>g((x,y'))$ if~$y>y'$. For a given function~$g:U\to \mathbb{R}$, 
		we define the family of orderings~$\succ_g$ as~$p\succ_g q$ if~$g(p)>g(q)$ for all~$p,q\in	U$.
		For two points~$p$ and~$q$ we say that~$p$ \emph{dominates}~$q$ if~$x(p)\ge x(q)$ 
		and~$y(p)\ge y(q)$ and at least one of the inequalities is strict.
		Note that the function that~$\succ$ corresponds to ($g((x,y))=x+y$) is both symmetric 
		and strictly increasing in~$x$ and~$y$. Therefore, the following theorem bounds 
		the approximation ratio of a family of \greedy{} algorithms that include the one 
		that we have discussed in the preceding~sections.

	\begin{figure}
		\centering
		\includegraphics[width=\linewidth]{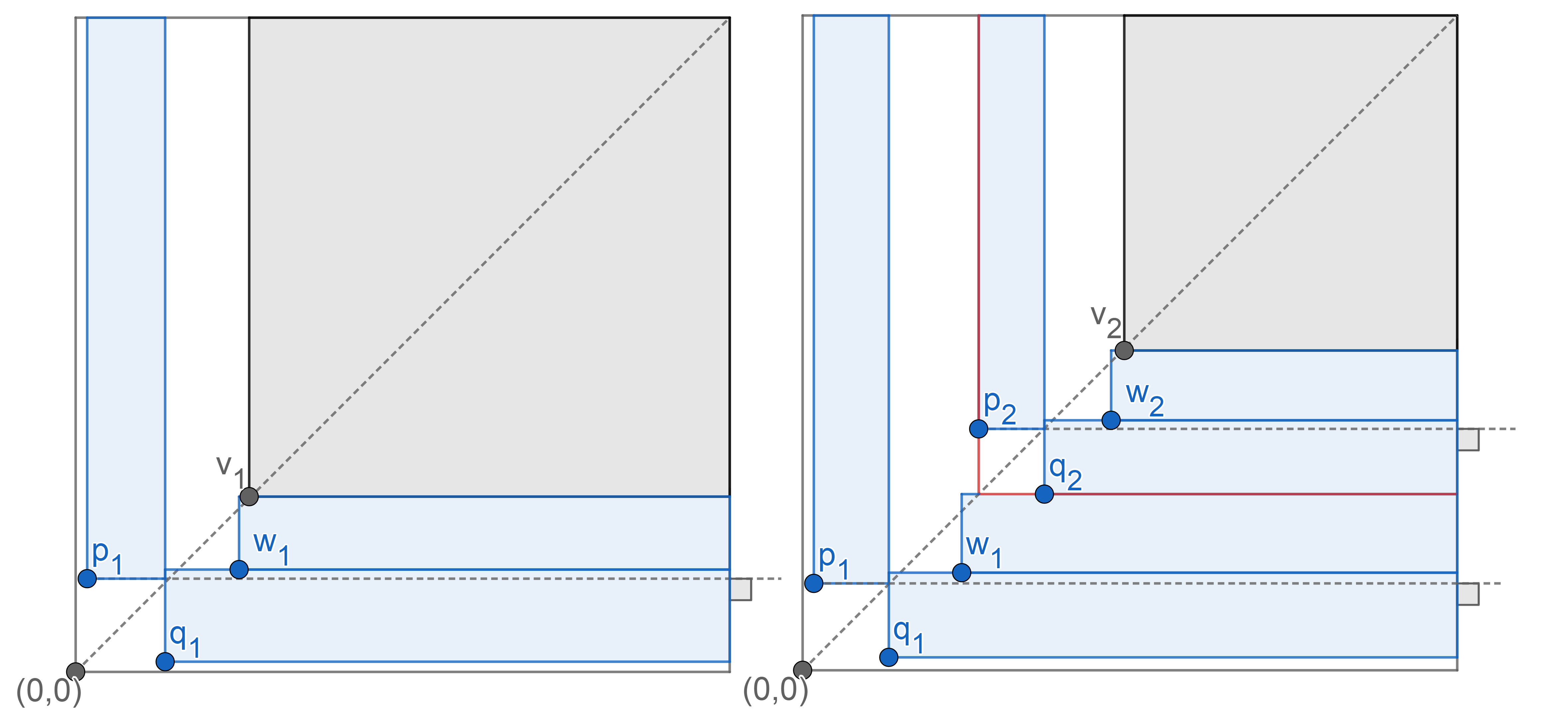}
		\caption{Left: the set~$S_{1,\eps}$, with the \greedy{} LLARP with $\succ_g$ ($r((0,0))$ is not drawn). Right: set~$S_{2,\eps}$. The choice of anchored rectangles is independent for the two `L-shapes', which are separated by the red lines.}
		\label{fig:approx_func}
	\end{figure}
	\begin{theorem} \label{thm:approx_func} 
		For any symmetric function~$g$ that is
		strictly increasing in both~$x$ and~$y$, the worst case approximation ratio for
		the \greedy{} algorithm with any ordering from~$\succ_g$ is at most~$\frac{3}{4}$.
	\end{theorem}
	\begin{proof}
		We construct sets~$S_{n,\eps}$ for positive integers~$n$ and small~$\eps>0$. 
		We show that, for these sets, as~$n\to\infty$ and~$\eps\to 0$, the optimal 
		covered area approaches $1$ and the area cover by a \greedy{} algorithm 
		that satisfies the conditions approaches~$\frac{3}{4}$. 
		
		We define~$S_{n,\eps}$ recursively. The set~$S_{n,\eps}$ consists of~$4n+1$ points, 
		\[S_{n,\eps}=\{(0,0),p_1,\ldots,p_n,q_1,\ldots,q_n,v_1,\ldots,v_n,w_1,\ldots,w_n\}\,,\]
		where the coordinates 
		of each point (except the origin) depend on $\eps$.
		
		Before we define the recursive relation, we first look at~$S_{1,\eps}$.
		Let~$p_1=(0,\eps-\eps^3)$,
		let~$q_1=(\eps,0)$, let~$v_1=(2\eps,2\eps)$, and let~$w_1=(2\eps,\eps)$. 
		See also the left side of Figure~\ref{fig:approx_func}.
		We see that~$v_1$ dominates all other points and~$w_1$ dominates all other
		points except~$v_1$. Furthermore, the reflection of~$p_1$ in the line~$x=y$
		is~$p_1'=(\eps-\eps^3,0)$ and is dominated by~$q_1$.
		Since~$g$ is symmetric and strictly increasing,~$q_1\succ_g p_1$ and $q_1$ 
		is treated 
		before~$p_1$ by the algorithm. 
		Thus~$v_1\succ_g w_1\succ_g q_1\succ_g p_1\succ_g (0,0)$.  Denote the rectangle 
		anchored at~$p$ that the \greedy{} algorithm adds by~$r(p)$. Consider the
		options for~$r(q_1)$, there are two.
		The rectangle that is bounded by~$w_1$ and the right edge of~$U$, 
		and the one that is bounded by~$w_1$ and the top
		edge of~$U$. The former has an area of size 
		\[
			\left(1-x(q_1)\right)\left(y(w_1)-y(q_1)\right)= \left(1-\eps\right)\eps=\eps-\eps^2
		\] 	
		and the latter has an area of size 
		\[ 
			\left(x(w_1)-x(q_1)\right)\left(1-y(q_1)\right)=
			\left(2\eps-\eps\right)=\eps\,.
		\] 
		So both have an area of~$\eps-O(\eps^2)$, where we disregard larger powers of~$\eps$ 
		in the big-O notation. Moreover, independent of the rectangle 
		chosen for~$q_1$, for~$p_1$ the largest-area rectangle is bounded by the upper edge of~$U$ 
		and the left edge of~$r(q_1)$ and has an area of size
		\[
			\left(x(q_1)-x(p_1)\right)\left(1-y(p_1)\right) = 
			\eps\left(1-(\eps-\eps^3)\right) =
			\eps-\eps^2+\eps^4\,.
		\]
		For~$w_1$ the largest-area rectangle is bounded by~$v_1$ and the right edge of~$U$
		and has an area of size
		\[ 
			\left(1-x(w_1)\right)\left(y(v_1)-y(w_1)\right)=
			\left(1-2\eps\right)\left(2\eps-\eps\right)=
			\eps-2\eps^2\,.
		\] 
		Again, both~$r(p_1)$ and~$r(w_1)$ have an area of~$\eps-O(\eps^2)$. Moreover, the 
		rectangle~$r((0,0))$ has area~$O(\eps^2)$ and thus the total area covered 
		by~$r((0,0))$,~$r(p_1)$,~$r(q_1)$, and~$r(w_1)$ is~$3\eps-O(\eps^2)$.
		
		There exists a solution that covers~$4\eps-O(\eps^2)$, namely by 
		taking for~$q_1$ the rectangle with upper-right point~$(1,\eps-\eps^3)$
		and for~$p_1$ the rectangle with upper-right point~$(2\eps,1)$. 
		The resulting set of rectangles covers~$4\eps-O(\eps^2)$.
		
		By replacing~$v_1$ by the same pattern in the square between~$v_1$ and the edges of~$U$, 
		we obtain~$S_{2,\eps}$ and iteratively we obtain from~$S_{k,\eps}$ the set~$S_{k+1,\eps}$.
		Note that for any $S_{n,\eps}$ the choices that we have for the rectangles in 
		any two subsets $S_{k,\eps}$ and $S_{k+1,\eps}\setminus S_{k,\eps}$, with $k\le n-1$,
		are independent. Thus, our analysis for $S_{1,\eps}$ holds for all 
		repeated patterns. See the right side of Figure~\ref{fig:approx_func}.
		
		Now consider~$S_{n,\eps}$, as~$\eps\to0$ we know that the \greedy{} algorithm 
		covers at most~$3/4$ of what an optimal solution covers in the area~$U\setminus r(v_n)$. 
		What remains is to show that as~$n\to\infty$ we have $\area(r(v_n))\to0$. 
		
		The area of~$r(v_1)$ is clearly~$(1-2\eps)^2$. When we repeat the pattern 
		for~$r(v_2)$, we see that its area is~$((1-2\eps)^2)^2$. In general, we 
		get that
		\[
			\lim_{n\to\infty}\area(r(v_n)) = \lim_{n\to\infty}((1-2\eps)^2)^n 
			= \lim_{n\to\infty}(1-2\eps)^{2n} = 0\,,
		\] 
		which finishes our proof.
	\end{proof}

	\section{Concluding remarks}	
	
	When we compare our analysis for the lower bound to that of Dumitrescu and 
	Tóth~\cite{Dumitrescu2015}, 
	we observe that our 
	analysis allows us to bound the area for smaller beta ($\beta_0$ is that value 
	for which $F(\beta,\lambda,\alpha)$ is equal to~$1$). Therefore, improving the 
	analysis not only by having a better bound for each~$\beta$, but also allowing 
	us to integrate over slightly more values of~$\beta$ in the final~stage.

	We believe that the most important take-away from our result is that, 
	while we adapt the analysis in several ways that could be viewed as major, 
	the improvement of the lower bound is minor. To us, this indicates that, indeed,
	as Dumitrescu and Tóth~\cite{Dumitrescu2015} have conjectured before us,
	significant improvements will not come from adapting this analysis, but rather 
	from a different proof altogether. We hope that the insights that we give here 
	can inspire someone to close the gap. 
	
	What can still be improved in the current analysis is the optimization over 
	the variables~$\lambda$ and~$\alpha$. In particular, if an analytical 
	optimum of~$F(\beta,\lambda,\alpha)$ can be found for~$\lambda$ and~$\alpha$
	for fixed~$\beta$, these values can be used in the computation of~\eqref{eq:integral}.
	This should then result in a better bound since both~$\lambda$ and~$\alpha$ depend 
	on~$\beta$ within the integral. 
	
	With respect to an upper bound on the performance of the \greedy{} algorithm, 
	we now know have a first upper bound on the approximation ratio. Clearly, 
	this does not tell us much about the original conjecture, since for the simple 
	example that shows the absolute upper bound of $1/2$ (see, e.g.,~\cite{Dumitrescu2015})
	the optimal and \greedy{} solutions coincide. In this respect, we still view 
	the open question ``is Bob's problem NP-hard'' as the most interesting one. 
	Although it would be nice to see the approximation ratio gap closed, this 
	may prove to be equally difficult as proving Freedman's Conjecture.

	%---------------------------------------------------------------------
	\bibliographystyle{plainurl} 
	\bibliography{bib}
	%---------------------------------------------------------------------
	
%	
%
	\appendix

	\section{Analysis of an improved lower bound through continuous \texorpdfstring{$\beta$}{β}}\label{app:integral}
	In this section we show how to construct the lower bound on the total area 
	covered by the \tp{} algorithm by letting~$\beta$ run over all values that 
	positively contribute to the bound. Intuitively, this means that~$\beta$ 
	runs from the value such that we can bound the total area of~$\beta$-tile by 
	less than one up to infinity. This analysis follows that of Section~3.3 
	in~\cite{Dumitrescu2015}.
	
	We showed that the total area of~$\beta$-tiles is bounded by~$F(\beta,\lambda,\alpha)$.
	Thus providing a lower bound on what the \tp{} algorithm can cover of 
	\[
		\frac{1-F(\beta,\lambda,\alpha)}{\beta}\,.
	\]
	If we consider, instead of one value, two different values for~$\beta$, 
	say~$\beta_0$ and~$\beta_1$, such that~$\beta_0>\beta_1$ and realize that 
	the set of~$\beta_1$-tiles contains all~$\beta_0$ tiles as well, we get 
	an improved lower bound of 
	\[
		\frac{1-F(\beta_0,\lambda,\alpha)}{\beta_0} + 
		\frac{F(\beta_0,\lambda,\alpha)-F(\beta_1,\lambda,\alpha)}{\beta_1} \,.
	\]
	In the same fashion, we can improve the lower bound by considering more 
	and more different values for~$\beta$. For an arbitrary integer~$k>0$ 
	and values~$\beta_0,\ldots,\beta_k$, we obtain a lower bound of 
	\begin{align*}
		\area(R)&\ge \frac{1-F(\beta_0,\lambda,\alpha)}{\beta_0} + 
		\sum_{i=1}^k \frac{F(\beta_{i-1},\lambda,\alpha)
			-F(\beta_{i},\lambda,\alpha)}{\beta_i}\\
		&= \frac{1-F(\beta_0,\lambda,\alpha)}{\beta_0} - 
		\sum_{i=1}^k \frac1{\beta_i}\frac{F(\beta_i,\lambda,\alpha)
			-F(\beta_{i-1},\lambda,\alpha)}{\beta_i-\beta_{i-1}}
			\left(\beta_i-\beta_{i-1}\right) \,.
	\end{align*}
	In particular, we can set~$\beta_i=\beta_{i-1}+\eps$ to get 
	\[
		\area(R)\ge\frac{1-F(\beta_0,\lambda,\alpha)}{\beta_0} - 
		\sum_{i=1}^k \frac\eps{\beta_i}\frac{F(\beta_{i-1}+\eps,\lambda,\alpha)
			-F(\beta_{i-1},\lambda,\alpha)}{\eps}
		 \,.
	\]
	When we now let~$k\rightarrow\infty$ and~$\eps\rightarrow0$, we get 
	\[
		\area(R)\ge\frac{1-F(\beta_0,\lambda,\alpha)}{\beta_0} - 
		\int_{\beta_0}^\infty \frac1{\beta}\frac{\partial F(\beta,\lambda,\alpha)
			}{\partial \beta}\ud\beta
		\,.
	\]
	Substituting~\eqref{eq:area}, this is equal to
	\begin{align*}
		\area(R)&\ge\frac1{\beta_0} - \frac{(1+2\alpha-2\lambda)(2(1+\alpha)^2+\lambda(2+\alpha))e^{3+2\alpha}}
		{4\lambda(\alpha-\lambda)(1+\alpha)^2 }\left(\frac1{e^{\beta_0}} - 
		\int_{\beta_0}^\infty 
		\left(\frac1{\beta}\frac{\ud}{\ud\beta}\frac{\beta}{e^{\beta}}\right)\ud\beta\right)\\
		&= \frac1{\beta_0} - \frac{(1+2\alpha-2\lambda)(2(1+\alpha)^2+\lambda(2+\alpha))e^{3+2\alpha}}
		{4\lambda(\alpha-\lambda)(1+\alpha)^2 }\left(\frac1{e^{\beta_0}} + 
		\int_{\beta_0}^\infty 
		\left(\frac1{\beta}\frac{\beta-1}{e^{\beta}}\right)\ud\beta\right)\\
		&= \frac1{\beta_0} - \frac{(1+2\alpha-2\lambda)(2(1+\alpha)^2+\lambda(2+\alpha))e^{3+2\alpha}}
		{4\lambda(\alpha-\lambda)(1+\alpha)^2 }\left(\frac1{e^{\beta_0}} + 
		\int_{\beta_0}^\infty 
		\left(\frac1{\beta e^{\beta}}-\frac{1}{e^{\beta}}\right)\ud\beta\right)\\
		&= \frac1{\beta_0} - \frac{(1+2\alpha-2\lambda)(2(1+\alpha)^2+\lambda(2+\alpha))e^{3+2\alpha}}
		{4\lambda(\alpha-\lambda)(1+\alpha)^2 }
		\int_{\beta_0}^\infty 
		\frac{1}{\beta e^{\beta}}\ud\beta\,.
	\end{align*}

	\section{Matlab code}
	\subsection{Code for 0.0806}\label{app:0.08}
	\begin{verbatim}
		y = @(l,b,a)-(1./b-(1+2*a-2*l)./(2*a-2*l).*(2*(1+a).^2+l.*(2+a))./
		(1+a).^2./l*exp(3+2*a-b)./2);
		yx = @(x) y(x(1), x(2),x(3));
		A = [];
		Aeq = [];
		b = [];
		Beq = [];
		lb = [0.01 5 0.51];
		ub = [0.5 100 1.5];
		[x, fval] = fmincon(yx, [0.3, 6,1], A, b, Aeq, Beq, lb, ub);
		x
		fval
	\end{verbatim}
	
	\subsection{Code for 0.10390}\label{app:0.10}
	\begin{verbatim}
		y = @(l,b,a)-(1./b-(1+2*a-2*l)./(2*a-2*l).*(2*(1+a).^2+l.*(2+a))./
		(1+a).^2./l*exp(3+2*a)*expint(b)./2);
		yx = @(x) y(x(1), x(2),x(3));
		A = [];
		Aeq = [];
		b = [];
		Beq = [];
		lb = [0.01 5 0.51];
		ub = [0.5 100 1.5];
		[x, fval] = fmincon(yx, [0.3, 6,1], A, b, Aeq, Beq, lb, ub);
		x
		fval
	\end{verbatim}
\end{document}